\documentclass[10pt]{article}
\usepackage{amssymb}
\usepackage{fullpage}
\usepackage{amsmath,amsfonts,amssymb,amsthm}
\usepackage{mathrsfs}
\usepackage{algorithm}
\usepackage{algorithmic}
\usepackage{latexsym}
\usepackage[b5paper,textheight=18.1cm,textwidth=11.4cm,dvipdfm]{geometry}

\newtheorem{theorem}{Theorem}

\newtheorem{corollary}{Corollary}

\def \sd{\gamma_{s}}
\def \std{\gamma^{s}_{t}}
\def \md{\gamma^{-}}

\begin{document}
\title{Signed and Minus Domination in Complete Multipartite Graphs\footnotemark[1]}


\author{Hongyu Liang\footnotemark[2]}
\renewcommand{\thefootnote}{\fnsymbol{footnote}}


\footnotetext[2]{Institute for Interdisciplinary Information Sciences,
Tsinghua University, Beijing, China. E-mail: lianghy08@mails.tsinghua.edu.cn}
\date{}
\maketitle

\begin{abstract}
In this paper we determine the exact values of the signed domination number, signed total domination number, and minus domination number of complete multipartite graphs, which substantially generalizes some previous results obtained for special subclasses of complete multipartite graphs such as cliques and complete bipartite graphs.
\end{abstract}

\section{Introduction}
\label{sec:intro} All graphs in this paper are simple and undirected. We generally follow
\cite{gt} for standard notation and terminologies in graph theory.
For a graph $G$, its vertex set and edge set are denoted by $V(G)$ and $E(G)$, respectively. For each $v\in V(G)$,
$N_G(v):=\{u\in V(G)~|~\{u,v\}\in E(G)\}$ is the \emph{open neighborhood} of $v$, and $N_G[v]:=N_G(v)\cup \{v\}$ is the \emph{closed neighborhood} of $v$.
For a function $f:V(G)\rightarrow \mathbb{R}$, define $f(S):=\sum_{v\in S}f(v)$ for all $S\subseteq V(G)$; the \emph{weight} of $f$ is $w(f):=f(V(G))$.


Given a graph $G$, a function $f:V(G)\rightarrow \{-1,1\}$ is called a \emph{signed dominating function} (respectively, \emph{signed total dominating function}) of $G$ if $f(N_G[v])\geq 1$ (respectively, $f(N_G(v))\geq 1$) for all $v\in V(G)$. The \emph{signed domination number} (respectively, \emph{signed total domination number}) of $G$, denoted by $\gamma_s(G)$ (respectively, $\std(G)$), is the minimum weight of a signed dominating function (respectively, signed total dominating function) of $G$. Notice that the signed total domination number is only defined for graphs without isolated vertices. The notion of signed domination and signed total domination have been extensively studied in the literature; see, e.g., \cite{signed_cm96,signed_f96,signed_hw02,h04,signed_m00,ac_xsc05,std_cmj01,md_cmj06} and the references therein.
For a graph $G$, a function $f:V(G)\rightarrow \{-1,0,1\}$ is called a \emph{minus dominating function} of $G$ if $f(N_G[v])\geq 1$ for every $v\in V(G)$. The \emph{minus domination number} of $G$, denoted by $\md(G)$, is the minimum weight of a minus dominating function of $G$. Minus domination has been studied in, e.g., \cite{md_dam01,md_dam96,md_dm96,md_dm99,md_cmj06}.
For a comprehensive treatment on the theory of domination in graphs, the reader is referred to \cite{dom_book2,dom_book}.

The exact values of the signed domination number, signed total domination number, and minus domination number have been determined for some special classes of graphs including complete graphs and complete bipartite graphs. To our knowledge, however, the values of these parameters in a more general class of graphs, namely the class of complete multipartite graphs, have not been decided yet. In this paper we fill this gap by completely determining the values of the three parameters in complete multipartite graphs. Our work substantially generalizes the previously obtained results for complete graphs (note that a complete graph of order $n$ is also a complete $n$-partite graph) and complete bipartite graphs.

\section{Signed (Total) Domination and Minus Domination in Complete Multipartite Graphs}
Let $k\geq 2$ and $n_1,n_2,\ldots,n_k$ be positive integers.
Throughout the paper, $K_{n_1,n_2,\ldots,n_k}$ denotes the complete $k$-partite graph with vertex set $V=V_1\cup V_2\cup\ldots \cup V_k$ and edge set $E$, where $V_i=\{v_{i,j}~|~1\leq j\leq n_i\}$ and $E=\{\{v_{i,j},v_{i',j'}\}~|~i\neq i'; 1\leq j\leq n_i; 1\leq j'\leq n_{i'}\}$.
Let $t$ be the number of $i$'s for which $n_i$ is odd; that is, $t=|\{i~|~1\leq i\leq k; n_i\equiv 1 \pmod 2\}|$. Assume without loss of generality that $n_1,n_2,\ldots,n_{t}$ are odd, whereas $n_{t+1},\ldots,n_{k}$ are even.
Let $I_1=\{i~|~1\leq i\leq t; n_i=1\}$ and $I_{2}=\{i~|~t+1\leq i\leq k; n_i=2\}$.


We first consider the signed domination number.
\begin{theorem}\label{thm:signed}
If $t$ is odd, then
\begin{eqnarray*}
\sd(K_{n_1,n_2,\ldots,n_k})=\left\{
\begin{array}{ll}
1 & \textrm{if~~}t\geq 3 \textrm{~and~} |I_1|\geq \frac{2k-t+1}{2};\\
1+n_2 & \textrm{if~~}t=1, k=2, \textrm{~and~}n_1=1;\\
5 & \textrm{if~~}(t=1, k=2, n_1\geq 5, n_2\geq 4) \textrm{~or~} \\
& ~~~~(t=1, k\geq 3, n_1\neq 3, (\forall 2\leq i\leq k) n_i\geq 4);\\
3 & \textrm{otherwise}.
\end{array}
\right.
\end{eqnarray*}
If $t$ is even, then
\begin{eqnarray*}
\sd(K_{n_1,n_2,\ldots,n_k})=\left\{
\begin{array}{ll}
n_1+n_2 & \textrm{if~~}t=k=2 \textrm{~and~} \min\{n_1,n_2\}=1;\\
6 & \textrm{if~~}t=k=2 \textrm{~and~} \min\{n_1,n_2\}\geq 5;\\
2 & \textrm{if~~}|I_1|+|I_2|\geq \frac{t}{2}+1;\\
4 & \textrm{otherwise}.
\end{array}
\right.
\end{eqnarray*}
\end{theorem}

\begin{proof}
Let $G=K_{n_1,n_2,\ldots,n_k}$ and $f$ be a signed dominating function of $G$ with $w(f)=\sd(G)$.
We first give some observations that will be frequently used in the proof.
For every $1\leq i\leq k$ and $1\leq j\leq n_i$, we have
\begin{equation}\label{equ:neib}
f(N_G[v_{i,j}])=f(v_{i,j})+f(V\setminus V_{i})\geq 1,
\end{equation}
and consequently,
\begin{equation}\label{equ:geq0}
f(V\setminus V_{i})\geq 0.
\end{equation}
Therefore,
\begin{equation*}
\sd(G)=w(f)=f(V)=\frac{1}{k-1}\sum_{i=1}^{k}f(V\setminus V_{i})\geq 0.
\end{equation*}
Thus $w(f)=0$ holds only if $f(V\setminus V_i)=0$ for all $i$. By (\ref{equ:neib}) this implies $f(v_{i,j})=1$ for all $i,j$ and thus $w(f)>0$, a contradiction! Therefore $w(f)\geq 1$. As $w(f)\equiv |V|\equiv t \pmod 2$, we have
\begin{equation}\label{equ:lb}
\sd(G)\geq 1 \textrm{~if~}t \textrm{~is odd, and~}\sd(G)\geq 2 \textrm{~if~}t\textrm{~is even.}
\end{equation}

We now turn to the main part of the proof. First consider the case when $t$ is odd. We perform a case analysis as follows.
\begin{enumerate}
\item $t=1$ and $k=2$. That is, $G=K_{n_1,n_2}$ where $n_1$ is odd and $n_2$ is even. Applying Theorem 1 from \cite{md_cmj06} under different situations gives the following:
\begin{itemize}
\item When $n_1=1$, $\sd(G)=n_2+1$.
\item When $n_1=3$, $\sd(G)=3$.
\item When $n_1\geq 5$ and $n_2=2$, $\sd(G)=3$.
\item When $n_1\geq 5$ and $n_2\geq 4$, $\sd(G)=5$.
\end{itemize}

\item $t=1$ and $k\geq 3$. We first show that $\sd(G)=w(f)\geq 3$. Assume to the contrary that $w(f)<3$. Since $G$ has odd number of vertices, $w(f)$ should be odd, and thus $w(f)=1$ by (\ref{equ:lb}). Fix $i'\in\{2,3,\ldots,k\}$. By (\ref{equ:geq0}) and the fact that $|V_{i'}|\geq 2$, there exists $1\leq j\leq n_{i'}$ for which $f(v_{i',j})=-1$, otherwise $w(f)\geq 2$. This, by (\ref{equ:neib}), indicates that $f(V\setminus V_{i'})\geq 2$. Noting that $|V\setminus V_{i'}|$ is odd, we have $f(V\setminus V_{i'})\geq 3$. Due to the arbitrariess of $i'$, we obtain:
    \begin{equation}\label{equ:3}
    w(f)=\frac{1}{k-1}\sum_{i=1}^{k}f(V\setminus V_{i})\geq \frac{0+3(k-1)}{k-1}=3,
    \end{equation}
    contradicting with our previous assumption of $w(f)<3$. Therefore we have $\sd(G)=w(f)\geq 3$.

    Furthermore, we will prove that $w(f)\geq 5$ if, in addition, $n_1\neq 3$ and $(\forall 2\leq i\leq k)n_i\geq 4$. Suppose to the contrary that $w(f)\leq 3$ in this case. Since $w(f)\geq 3$, we have $w(f)=3$. Analogously to the previous analysis, for each $2\leq i\leq k$ there is $1\leq j\leq n_i$ such that $f(v_{i,j})=-1$, which implies $f(V\setminus V_i)\geq 3$. Thus (\ref{equ:3}) still holds. To achieve the equality, we must have $f(V\setminus V_1)=0$ and $f(V\setminus V_i)=3$ for all $2\leq i\leq k$. This, by (\ref{equ:neib}), implies that $f(v_{1,j})=1$ for all $1\leq j\leq n_1$. Therefore $3=w(f)=f(V_1)=n_1$, which contradicts with our assumption that $n_1\neq 3$. Thus, we have established that $w(f)\geq 5$ when $n_1\neq 3$ and $(\forall 2\leq i\leq k)n_i\geq 4$.

    We next prove that these lower bounds are attainable in respective cases. Consider the function $f':V\rightarrow \{-1,1\}$ defined as follows: Assign $+1$ to $\frac{n_1+1}{2}$ vertices in $V_1$, to $\frac{n_i+2}{2}$ vertices in $V_i$ for $i\in\{2,3\}$, and to $\frac{n_j}{2}$ vertices in $V_j$ for all $3\leq j\leq k$; assign $-1$ to all other vertices in $V$. It is easy to see that $f'(V_1)=1$, $f'(V_2)=f'(V_3)=2$, and $f'(V_j)=0$ for $4\leq j\leq k$. Evidently $f'$ is a signed dominating function of $G$, which has weight 5. Thus $\sd(G)\leq 5$, which is tight for the case where $n_1\neq 3$ and $(\forall 2\leq i\leq k)n_i\geq 4$.
    Now consider the case where $n_1=3$ or $(\exists 2\leq i\leq k)n_i=2$. If $n_1=3$, the function that assigns $+1$ to all vertices in $V_1$ and $\frac{n_i}{2}$ vertices in $V_i$ for all $2\leq i\leq k$, and $-1$ to all other vertices, is a signed dominating function of $G$ of weight 3. If $n_i=2$ for some $i\in\{2,3,\ldots,k\}$, we can construct a signed dominating function of $G$ of weight 3 by assigning $+1$ to $\frac{n_1+1}{2}$ vertices in $V_1$, to both vertices in $V_i$, to $\frac{n_j}{2}$ vertices in $V_j$ for all $j\in\{2,\ldots,k\}\setminus \{i\}$, and assigning $-1$ to all other vertices in $V$. Hence, $\sd(G)=3$ when $n_1=3$ or $n_i=2$ for some $2\leq i\leq k$. This finishes the analysis of the case where $t=1$ and $k\geq 3$.

\item $t\geq 3$.
    First assume that $|I_1|\leq \frac{2k-t-1}{2}$. (Recall that $I_1$ is the set of indices $i$ for which $n_i=1$.) We will show that $\sd(G)\geq 3$. Assume to the contrary that $\sd(G)\leq 1$. For each $1\leq i\leq k$ for which $n_i\geq 2$, by (\ref{equ:geq0}), there exists $1\leq j\leq n_i$ such that $f(v_{i,j})=-1$. According to (\ref{equ:neib}), $f(V)-f(V_i)=f(V\setminus V_i)\geq 2$, and thus $$f(V_i)\leq f(V)-2=\sd(G)-2\leq -1.$$
    When $t+1\leq i\leq k$, $|V_i|=n_i$ is even, and thus the above inequality can be improved to $f(V_i)\leq -2$. Noting that $f(V_j)\leq 1$ for all $j$ such that $n_j=1$, we have:
    \begin{eqnarray*}
    \sd(G)=f(V)\leq |I_1|-(t-|I_1|)-2(k-t)=2|I_1|-2k+t\leq -1,
    \end{eqnarray*}
    which, however, is a contradiction to (\ref{equ:lb}).
    As a consequence, our assumption that $\sd(G)\leq 1$ cannot hold, and thus $\sd(G)\geq 3$. On the other hand, consider the function $f'$ defined as follows: Assign $+1$ to $\frac{n_i+1}{2}$ vertices in $V_i$ for all $1\leq i\leq \frac{t+3}{2}$ (note that $\frac{t+3}{2}\leq t$ since $t\geq 3$), to $\frac{n_i-1}{2}$ vertices in $V_i$ for all $\frac{t+3}{2}< i\leq t$, and to $\frac{n_i}{2}$ vertices in $V_i$ for all $t<i\leq k$; assign $-1$ to all the other vertices in $V$. It is easy to verify that $f'$ is a signed dominating function of $G$ of weight $3$, and thus $\sd(G)\leq 3$. Therefore we have $\sd(G)=3$.

Now suppose $|I_1| \geq \frac{2k-t+1}{2}$. (Note that this implies $t\geq \frac{2k-t+1}{2}$.) Without loss of generality we assume $n_1=n_2=\ldots=n_{\frac{2k-t+1}{2}}=1$. Define a function $f':V\rightarrow \{-1,1\}$ as follows: Assign $+1$ to all the vertices in $\bigcup_{i=1}^{\frac{2k-t+1}{2}}V_i=\{v_{i,1}~|~1\leq i\leq \frac{2k-t+1}{2}\}$, to $\frac{n_i-1}{2}$ vertices in $V_i$ for all $\frac{2k-t+1}{2}<i\leq t$, and to $\frac{n_i-2}{2}$ vertices in $V_i$ for all $t+1\leq i\leq k$; assign $-1$ to all other vertices in $V$. It is easy to verify that $f'(V_i)=1$ for $1\leq i\leq \frac{2k-t+1}{2}, f'(V_i)=-1$ for $\frac{2k-t+1}{2}<i\leq t$, and $f'(V_i)=-2$ for $t<i\leq k$. The weight of $f'$ is
$$f'(V)=\frac{2k-t+1}{2}-(t-\frac{2k-t+1}{2})-2(k-t)=1.$$
For each $1\leq i\leq \frac{2k-t+1}{2}$, $f'(N_G[v_{i,1}])=f'(V)=1$ (recall that $n_i=1$). For each $i>\frac{2k-t+1}{2}$, since $f'(V_i)\leq -1$, we have that for each $1\leq j\leq n_i$, $f'(N_G[v_{i,j}])=f'(v_{i,j})+(f'(V)-f'(V_i))\geq f'(V)=1$. Therefore, $f'$ is a signed dominating function of $G$, implying that $\sd(G)\leq w(f')=1$. By (\ref{equ:lb}), $\sd(G)=1$.
This completes the whole analysis for the case where $t$ is odd.
\end{enumerate}

We next turn to the situation where $t$ is even. We will prove the following:
\begin{equation}\label{equ:even}
\sd(G)\geq 4 \textrm{~~if~~} |I_1|+|I_2|\leq \frac{t}{2}\;.
\end{equation}

Assume that $\sd(G)\leq 2$ while $|I_1|+|I_2|\leq \frac{t}{2}$. By (\ref{equ:lb}) we have $\sd(G)=2$.
For each $1\leq i\leq k$ for which $n_i\geq 3$, by (\ref{equ:geq0}), there exists $1\leq j\leq n_i$ such that $f(v_{i,j})=-1$. According to (\ref{equ:neib}), $f(V\setminus V_i)=f(V)-f(V_i)\geq 2$, and thus
$$f(V_i)\leq f(V)-2=\sd(G)-2=0.$$
    When $n_i$ is odd, the above inequality can be improved to $f(V_i)\leq -1$. Therefore,
    \begin{eqnarray*}
    \sd(G)=w(f)\leq |I_1|-(t-|I_1|)+2|I_2|=2(|I_1|+|I_2|)-t\leq 0,
    \end{eqnarray*}
    contradicting with our assumption that $\sd(G)\leq 2$. Thus the inequality (\ref{equ:even}) is proved.
    We next consider several cases.

\begin{enumerate}
\item $|I_1|+|I_2|\geq \frac{t}{2}+1$. Choose two integers $i_1,i_2$ such that $0\leq i_1\leq |I_1|, 0\leq i_2\leq |I_2|$, and $i_1+i_2=\frac{t}{2}+1$. (Obviously such $i_1,i_2$ exist.) Without loss of generality, we assume that $n_1=n_2=\ldots=n_{i_1}=1$ and $n_{t+1}=n_{t+2}=\ldots=n_{t+i_2}=2$.
Consider the function $f':V\rightarrow \{-1,1\}$ obtained as follows: Assign $+1$ to $\frac{n_i+1}{2}$ vertices in $V_i$ for all $1\leq i\leq i_1$, to $\frac{n_i-1}{2}$ vertices in $V_i$ for all $i_1+1\leq i\leq t$, to $\frac{n_i+2}{2}$ vertices in $V_i$ for all $t+1\leq i\leq t+i_2$, and to $\frac{n_i}{2}$ vertices in $V_i$ for all $t+i_2+1\leq i\leq k$; assign $-1$ to all other vertices in $V$. Thus, $f'(V_i)=1$ for $1\leq i\leq i_1$, $f'(V_i)=-1$ for $i_1+1\leq i\leq t$, $f'(V_i)=2$ for $t+1\leq i\leq t+i_2$, and $f'(V_i)=0$ for $t+i_2<i\leq k$. So,
$$w(f')=i_1-(t-i_1)+2i_2=2(i_1+i_2)-t=2.$$
We now verify that $f'$ is a signed dominating function of $G$.
For each $1\leq i\leq i_1$ and the (only) vertex $v_{i,1}\in V_i$, $f'(N_G[v_{i,1}])=f'(V)=w(f)=2$. For each $i_1+1\leq i\leq t$ and $1\leq j\leq n_i$, $f'(N_G[v_{i,j}])=f'(v_{i,j})+f'(V)-f'(V_i)\geq -1+2-(-1)=3$. For each $t+1\leq i\leq t+i_2$ and $1\leq j\leq n_i(=2)$, $f'(N_G[v_{i,j}])=f'(v_{i,j})+f'(V)-f'(V_i)=1+2-2=1$. Finally, for each $t+i_2+1\leq i\leq k$ and $1\leq j\leq n_i$, $f'(N_G[v_{i,j}])=f'(v_{i,j})+f'(V)-f'(V_i)\geq -1+2=1$. Therefore, $f'$ is a signed dominating function of $G$, and thus $\sd(G)\leq w(f')=2$. By (\ref{equ:lb}) we have $\sd(G)=2$.

\item $|I_1|+|I_2|\leq \frac{t}{2}$.
\begin{enumerate}
\item $t=0$. In this case we can construct a signed dominating function of $G$ of weight 4 as follows: Assign $+1$ to $\frac{n_i+2}{2}$ vertices in $V_i$ for $i\in\{1,2\}$ and $\frac{n_i}{2}$ vertices in $V_i$ for $3\leq i\leq k$; assign $-1$ to all other vertices in $V$. Thus $\sd(G)\leq 4$, and by (\ref{equ:even}) we have $\sd(G)=4$.

\item $t=2$ and $k=2$. By applying Theorem 1 in \cite{md_cmj06} we obtain that
\begin{eqnarray*}
\sd(G)=\left\{
\begin{array}{ll}
n_1+n_2 & \textrm{~~if~~}\min\{n_1,n_2\}=1;\\
4 & \textrm{~~if~~}\min\{n_1,n_2\}=3;\\
6 & \textrm{~~if~~}\min\{n_1,n_2\}\geq 5.
\end{array}
\right.
\end{eqnarray*}
(Note that the condition $|I_1|+|I_2|\leq \frac{t}{2}$ excludes the situation $n_1=n_2=1$, in which case $\sd(G)=2$; nonetheless, this is compatible with the formula $n_1+n_2$.)

\item $t=2$ and $k\geq 3$. Consider the function $f':V\rightarrow \{-1,1\}$ defined as follows: Assign $+1$ to $\frac{n_i+1}{2}$ vertices in $V_i$ for $i\in\{1,2\}$, to $\frac{n_3+2}{2}$ vertices in $V_3$, and to $\frac{n_i}{2}$ vertices in $V_i$ for all $4\leq i\leq k$; assign $-1$ to all other vertices in $V$. Thus $f'(V_1)=f'(V_2)=1$, $f'(V_3)=2$, and $f'(V_i)=0$ for $i\geq 4$. It is easy to verify that $f'$ is a signed dominating function of $G$ of weight 4, and hence $\sd(G)\leq 4$. By (\ref{equ:even}) we have $\sd(G)=4$.

\item $t\geq 4$. Consider the function $f':V\rightarrow \{-1,1\}$ obtained as follows: Assign $+1$ to $\frac{n_i+1}{2}$ vertices in $V_i$ for all $1\leq i\leq \frac{t+4}{2}$, to $\frac{n_i-1}{2}$ vertices in $V_i$ for all $\frac{t+4}{2}<i\leq t$, and to $\frac{n_i}{2}$ vertices in $V_i$ for all $t< i\leq k$; assign $-1$ to all other vertices in $V$. Thus $f'(V_i)=1$ for $1\leq i\leq \frac{t+4}{2}$, $f'(V_i)=-1$ for $\frac{t+4}{2}<i\leq t$, and $f'(V_i)=0$ for $t+1\leq i\leq k$. It is easy to see that $f'$ is a signed dominating function of $G$ of weight 4, and thus $\sd(G)\leq 4$. By (\ref{equ:even}) we have $\sd(G)=4$. This also completes the whole analysis for the case $t$ is even.
\end{enumerate}
\end{enumerate}
The proof of Theorem~\ref{thm:signed} is thus completed.
\end{proof}

Our theorem generalizes Theorem 1 in \cite{md_cmj06}.
The following corollary is also immediate from it.
\begin{corollary}
When $\min\{n_1,n_2,\ldots,n_k\}\geq 2$, we have $1\leq \sd(K_{n_1,n_2,\ldots,n_k})\leq 6$.
\end{corollary}

We next deal with the signed total domination number.

\begin{theorem}\label{thm:kpartite}
\begin{eqnarray*}
\std(K_{n_1,n_2,\ldots,n_k})=\left\{
\begin{array}{ll}
3 & \textrm{~~if~~}t~\textrm{is odd};\\
4 & \textrm{~~if~~}t=0;\\
2 & \textrm{~~otherwise}.
\end{array}
\right.
\end{eqnarray*}
\end{theorem}

\begin{proof}
Let $G=K_{n_1,n_2,\ldots,n_k}$ and $f$ be a signed total dominating function of $G$ with $w(f)=\std(G)$.
Observe that $f(N(v_{i,1}))=f(V\setminus V_i)$ for each $i\in\{1,2,\ldots,k\}$. We investigate the following cases.
\begin{enumerate}
\item $t$ is odd. For every $i\in\{t+1,\ldots,k\}$, we have $f(V\setminus V_i)=f(N(v_{i,1}))\geq 1.$
For every $i\in\{1,2,\ldots, t\}$, since $|V\setminus V_i|$ is even, we have
$f(V\setminus V_i)\geq 2.$
Summing up the $k$ inequalities for $i=1,2,\ldots,k$ and noting that $\std(G)=f(V)=\frac{1}{k-1}\sum_{i=1}^{k}f(V\setminus V_i)$, we obtain that
$$\std(G)\geq \frac{(k-t)+2t}{k-1}= \frac{k+t}{k-1}>1.$$

As $|V|$ is odd, there is $\std(G)\geq 3$. We now prove that $\std(G)\leq 3$. Consider two further subcases:

\begin{enumerate}
\item $t=1$. Define a function $f':V\rightarrow \{-1,1\}$ by assigning
$+1$ to $\frac{n_1+1}{2}$ vertices in $V_1$, $\frac{n_2+2}{2}$ vertices in $V_2$, $\frac{n_i}{2}$ vertices in $V_i$ for all $3\leq i\leq k$, and assigning $-1$ to all other vertices in $V$.
It is easy to verify that $f'(V_1)=1$, $f'(V_2)=2$, and $f'(V_i)=0$ for $3\leq i\leq k$. Therefore, $f'$ is a signed total dominating function of $G$ of weight 3, and hence $\std(G)\leq 3$.

\item $t\geq 3$.
Define a function $f':V\rightarrow \{-1,1\}$ as follows:
Assign $+1$ to $\frac{n_i-1}{2}$ vertices in $V_i$ for all $1\leq i\leq \frac{t-3}{2}$, to $\frac{n_i+1}{2}$ vertices in $V_i$ for all $\frac{t-3}{2}<i\leq t$, and to $\frac{n_i}{2}$ vertices in $V_i$ for all $t+1\leq i\leq k$; assign $-1$ to all other vertices in $V$. Then, $f'(V_i)=0$ for all $t+1\leq i\leq k$, and among the $t$ values $f'(V_1),\ldots,f'(V_t)$, exactly $\frac{t-3}{2}$ of them are $-1$ and the others are all $+1$. It is thus easy to check that $f'$ is a signed total dominating function of $G$ of weight 3.
Hence $\std(G)\leq 3$.
\end{enumerate}

Combining parts (a) and (b), we have shown that $\std(G)\leq 3$, and thus $\std(G)=3$ when $t$ is odd.

\item $t$ is even.
For every $i\in\{1,2,\ldots, t\}$,
$f(V\setminus V_i)=f(N_G(v_{i,1}))\geq 1.$
For every $i\in\{t+1,\ldots,k\}$, since $|V\setminus V_i|$ is even, we have
$f(V\setminus V_i)\geq 2.$
Summing up the $k$ inequalities for $i=1,\ldots,k$, we obtain:
\begin{equation}\label{equ:1}
\std(G)=f(V)=\frac{1}{k-1}\sum_{i=1}^{k}f(V\setminus V_i)\geq \frac{2k-t}{k-1}>1,
\end{equation}
and thus $\std(G)\geq 2$.


Consider the following two subcases:
\begin{enumerate}
\item $t\geq 2$. Define a function $f':V\rightarrow \{-1,1\}$ as follows:
Assign $+1$ to $\frac{n_i-1}{2}$ vertices in $V_i$ for all $1\leq i\leq \frac{t-2}{2}$, to $\frac{n_i+1}{2}$ vertices in $V_i$ for all $\frac{t-2}{2}<i\leq t$, and to $\frac{n_i}{2}$ vertices in $V_i$ for all $t+1\leq i\leq k$. Assign $-1$ to all other vertices in $V$. Then, $f'(V_i)=0$ for all $t+1\leq i\leq k$, and among the $t$ values $f'(V_1),\ldots,f'(V_t)$, exactly $\frac{t-2}{2}$ of them are $-1$ and the others are all $+1$. It is easy to verify to $f'$ is a signed total dominating function of $G$ of weight 2, implying $\std(G)\leq 2$. Since $\std(G)\geq 2$ by (\ref{equ:1}), we have $\std(G)=2$.

\item $t=0$. Due to (\ref{equ:1}) we have $\std(f)\geq \frac{2k}{k-1}>2$. Since $|V|$ is even, it holds that $\std(f)\geq 4$. A signed total dominating function of $G$ of weight 4 can be obtained by assigning $+1$ to $\frac{n_i+2}{2}$ vertices in $V_i$ for all $i\in\{1,2\}$ and $\frac{n_i}{2}$ vertices in $V_i$ for all $3\leq i\leq k$, and assigning $-1$ to all other vertices in $V$. Thus $\std(G)=4$.
\end{enumerate}
\end{enumerate}

The proof of Theorem~\ref{thm:kpartite} is thus completed.
\end{proof}

Theorem~\ref{thm:kpartite} generalizes Propositions 1 and 4 in \cite{std_cmj01}. (We remark that Proposition 1 in \cite{std_cmj01} has a mistake: It should be that $\std(K_n)=3$ when $n$ is odd and at least 3.)

Finally we turn to the case of minus domination.

\begin{theorem}\label{thm:minus}
\begin{eqnarray*}
\md(K_{n_1,n_2,\ldots,n_k})=\left\{
\begin{array}{ll}
1 & \textrm{~~if~~}n_i=1\textrm{~for some~}i\in\{1,2,\ldots,k\};\\
2 & \textrm{~~otherwise}.\\
\end{array}
\right.
\end{eqnarray*}
\end{theorem}
\begin{proof}
Let $G=K_{n_1,n_2,\ldots,n_k}$ and $f$ be a minus dominating function of $G$ of weight $\md(G)$. If $n_i=1$ for some $i\in\{1,2,\ldots,k\}$, then $\md(G)=f(V)=f(N[v_{i,1}])\geq 1$. On the other hand, the function $f^*$ defined by $f^*(v_{i,1})=1$ and $f^*(v)=0$ for all $v\in V\setminus \{v_{i,1}\}$ is a minus dominating function of $G$ of weight 1. Therefore, $\md(G)=1$ in this case. We will assume in what follows that $n_i\geq 2$ for all $1\leq i\leq k$.

First, observe that the function $f^*$, defined by $f^*(v_{1,1})=f^*(v_{2,1})=1$ and $f^*(v)=0$ for all $v\in V\setminus \{v_{1,1},v_{2,1}\}$, is a minus dominating function of $G$ of weight 2. Thus $\md(G)\leq 2$. We assume that $\md(G)\leq 1$, and thus $f(V)\leq 1$. Fix an $i\in\{1,2,\ldots,k\}$. We have $f(v_{i,1})+f(V\setminus V_i)=f(N_G[v_{i,1}]))\geq 1$, and hence $f(V\setminus V_i)\geq 0$. Since $n_i\geq 2$ and $f(V)\leq 1$, there exists $v\in V_i$ such that $f(v)\leq 0$. Therefore, $f(V\setminus V_i)=f(N_G[v])-f(v)\geq 1$. Due to the arbitrariness of $i$, we get:
\begin{eqnarray*}
f(V)=\frac{1}{k-1}\sum_{i=1}^{k}f(V\setminus V_i)\geq \frac{k}{k-1}>1,
\end{eqnarray*}
contradicting with the fact that $f(V)\leq 1$. Thus, we have $\md(G)\geq 2$. Since we have proved $\md(G)\leq 2$ before, it holds that $\md(G)=2$. This completes the proof of Theorem~\ref{thm:minus}.
\end{proof}

Theorem~\ref{thm:minus} generalizes Theorem 1 in \cite{md_dm96} and Theorem 2 in \cite{md_cmj06}.

\section*{Acknowledgements}
This work was supported in part by the
National Basic Research Program of China Grant 2007CB807900,
2007CB807901, and the National Natural Science Foundation of China Grant
61033001, 61061130540, 61073174.

\bibliographystyle{plain}
\bibliography{multi}

\end{document}